\documentclass[aps,pre,reprint,superscriptaddress]{revtex4-1} 
\usepackage{amsmath,amssymb,graphicx,bm,natbib} 
\pdfoutput=1
\DeclareGraphicsExtensions{.pdf} 
\usepackage{enumitem}
\usepackage{tabularx}
\usepackage{csquotes}
\usepackage{amsthm,amsfonts}
\usepackage{siunitx}                % Typeset units

\theoremstyle{definition}
\newtheorem{definition}{Definition}
\newtheorem{algorithm}{Algorithm}
\theoremstyle{plain}
\newtheorem{theorem}{Theorem}

\renewcommand{\vec}[1]{\bm{#1}}

\def \Xpp {\vec{X}}
\def \F {\mathcal{F}}
\def \P {\mathbb{P}}
\def \Ns {\mathcal{N}}
\def \N {\mathbb{N}}
\def \R {\mathbb{R}}
\def \B {\mathcal{B}}
\def \f {\mathrm{f}}
\def \M {\mathcal{M}}
\def \C {\mathcal{C}}
\def \d {d}
\def \dist {\mathrm{dist}}

\usepackage{hyperref}

\begin{document} 

\title{Generation of Spatially Embedded Random Networks to Model Complex Transportation Networks}
\author{J\"urgen Hackl}
\email{hackl@ibi.baug.ethz.ch}
\affiliation{Institute for Construction Engineering and Management, Swiss Federal Institute of Technology, Zurich, Switzerland}
\author{Bryan T. Adey} 
\affiliation{Institute for Construction Engineering and Management, Swiss Federal Institute of Technology, Zurich, Switzerland} 
\date{\today} 

\begin{abstract} 
Random networks are increasingly used to analyse complex transportation networks, such as airline routes, roads and rail networks. So far, this research has been focused on describing the properties of the networks with the help of random networks, often without considering their spatial properties. In this article, a methodology is proposed to create random networks conserving their spatial properties. The produced random networks are not intended to be an accurate model of the real-world network being investigated, but are to be used to gain insight into the functioning of the network taking into consideration its spatial properties, which has potential to be useful in many types of analysis, e.g. estimating the network related risk. The proposed methodology combines a spatial non-homogeneous point process for vertex creation, which accounts for the spatial distribution of vertices, considering clustering effects of the network and a hybrid connection model for the edge creation. To illustrate the ability of the proposed methodology to be used to gain insight into a real world network, it is used to estimate standard structural statistics for part of the Swiss road network, and these are then compared with the known values.\\

\noindent\doi{10.1007/978-3-319-47886-9_15}
\end{abstract} 

%\pacs{05.45.-a, 89.75.-k, 89.75.Fb, 89.75.Hc} 
%Nonlinear dynamics and chaos 
%Complex systems 
%Structures and organization in complex systems 
%Networks and genealogical trees 

\maketitle 

\section{Introduction}
\label{sec:1}

Transportation networks are essential for economic growth and development. Over the past decades it has become obvious that the analysis and understanding of large-scale infrastructure networks is important for research, engineering and society. The failure or damage of an infrastructure system could cause huge social disruption. It could be out of all proportion to the actual physical damage \citep{Vespignani2010}. Thus, understanding the general principles, leading to the complex structures of these networks and their ability to withstand failures, natural hazards and man-made disasters, is critical for evaluating risk related to transportation networks, and for designing robust transportation networks to keep the risk within acceptable limits \citep{Schneider2011}.

In order to estimate this risk both the probability of different types of network failures and their associated consequences must be estimated. The estimation of the types of network failures is often estimated using percolation models in which some fraction of the total number of networks vertices or edges are removed. Using percolation models, as the number of vertices or edges are removed, the network is seen as undergoing transitions from the phase of connectivity (fully functional network) to the phase of dis-connectivity (non-functional network) \citep{Li2015}.

A good estimation of the type of failure requires a good estimation of how the failure of part of the network can spread, which is something that is dependent on the network topology and requires appropriate consideration of the mutual interplay between the structural complexity and functional dynamics of the network \citep{Kroger2011}. The estimation of the type of network failure is usually made using the exact network topology. Although good when possible, this approach is not always possible. For example, it is increasingly difficult with increasing network size and complexity, with increasing periods of time to be investigated, and increasing detail required \citep{Cadini2015}. It would be useful to have an easier, although less accurate way of being able to approximate the types of failures.

In this article, a model framework is proposed that has potential to be used for this purpose. The model framework can be used to create random networks that conserve the spatial properties of the networks to be analysed, and, therefore, allows more features of the transportation network to be captured then when using random networks developed using other model frameworks. The resulting random networks are, even if they are not accurate model of the real-world network being investigated, useful in gaining insights into the different types of network failure taking into consideration its spatial properties. The random networks developed taking into consideration their spatial properties are herein referred to as \emph{spatially embedded random networks}.

\section{Methodology}
\label{sec:2}

\subsection{Definition of a spatially embedded random network}
\label{sec:21}

Spatially embedded random networks can be seen as a special kind of random network, where vertices are placed randomly according to a specific distribution within a metric space. Edges are assigned to each pair of vertices with a probability, taking spatial properties into account.

\begin{definition}[Spatially embedded random network]
A random network $G=(V,E)$ defined on $S$ is called a spatially embedded random network if:
\begin{enumerate}[label={(\roman*)}]
\item $V$ is a mapping $\Xpp$ from a probability space $(\Omega_1,\F_1,\P_1)$ into a into a metric space $(N_\f,\Ns_\f)$. Where $\Omega_1$ is a sample space of all possible outcomes, $\F_1$ is a set of all considered events, and $\P_1$ is a probability measure of the events. Where $N_\f$ is defined as a space of finite subsets of $\R^d$ and $\Ns_\f$ is a Borel $\sigma$-algebra $\B$, such that for all bounded Borel sets $B \subseteq S$, the mapping $\Xpp\to N_{\Xpp}(B)$ is measurable.
\item $\M(\Xpp,g)$ is a connection model mapping from $\Omega = \Omega_1 \times \Omega_2$ into $N_\f \times \Omega_2$ defined by $(\omega_1,\omega_2) \to (\Xpp (\omega_1),\omega_2)$, where the sample space $\Omega_2$ is defined as $\prod_{\{K(n_u,z_u),K(n_v,z_v)\}}[0,1]$. The realization corresponding to $(\omega_1,\omega_2)$ is obtained for any two vertices $u$ and $v$ of $\Xpp(\omega_1)$, considering binary cubes. An edge $\{u,v\}$  exists if and only if $\omega_2 < g$  where $g$ is the spatially dependent connection function.
\end{enumerate}
\end{definition}

\cite{Barnett2007} pointed out that spatially embedded networks are essentially \enquote{ensemble-of-ensembles}, with the following two levels of randomness: the placement of the vertices, and an edge assignment, giving the coordinates of the vertices. Therefore, the properties of the graph $G$ are conditional on the vertex placement $\Xpp$.

In order to create such a spatially embedded network, a stochastic model for the spatially embedded vertices has to be introduced. As well as a connection model for the assignment of the edges, depending not only on the distances between the vertices but also other (spatial) properties.

\subsection{Vertex creation}
\label{sec:22}

The spatially embedded vertices are to be created using (inhomogeneous) Poisson point processes.
\begin{definition}[Poisson point process]
A Poisson point process with intensity function $\lambda$ is a point process $\Xpp$ on $S$, characterized by the following two properties:
\begin{enumerate}[label={(\roman*)}]
\item for any $B\subseteq S$, the number $N(B)$ of points in $B$ is a Poisson random variable with parameter $\lambda(B)$
\begin{equation}
  \P\{N(B)=k\} = \frac{\lambda(B)^k}{k!} \exp\{-\lambda(B)\} \qquad k\in\N_0,B\in S.
\end{equation}
\item if $N(B_1),\dots,N(B_n)$ are independent random variables for each $n\in\N$ and pairwise disjoint sets $B_1,\dots,B_n \in S$ with $\lambda(B_i)<\infty$.
\end{enumerate}
\end{definition}
If $\lambda$, the expected number of points per unit area, is constant, the Poisson process is called a \emph{homogeneous Poisson process} on $S$ with \emph{rate} or \emph{intensity} $\lambda$; otherwise it is said to be an \emph{inhomogeneous Poisson process} on $S$. Moreover, if $\lambda=1$ the process is called the \emph{standard Poisson point process} or \emph{unit rate Poisson process} on $S$.

If the average density of points is a function $\lambda(u)$ defined at all spatial locations $u$, then it is possible to fit an inhomogeneous Poisson process model to an observed covariate. Thereby, a covariate is a spatial function $Z(u)$ defined at all spatial locations $u$. For instance, the covariate $Z(u)$ might be the altitude or population density at location $u$.

\subsection{Edge creation}

The spatially embedded edges are to be created using a combination of a \emph{deterministic} and a \emph{random} connection model. In the first step, vertices are connected if they are close to each other (deterministic connection model). In the second step the edges are added, removed or rewired, by applying a random connection model. A parameter  is used to weight the two connection models. As $\kappa\to 0$, the connection model becomes a purely deterministic one. As $\kappa \to 1$, the connection model becomes purely a random one. More formally the connection model is:

\begin{definition}[Hybrid connection model]
\label{def:hybrid_connection_model}
Let $\Xpp$ be a finite point process on $\R^d$, $\M_d(\Xpp)$ a deterministic connection model, $\M_r(\Xpp,g)$ a random connection model with connection function $g$, and a parameter $\kappa$ satisfying $0\leq\kappa\leq 1$. To build the network
\begin{enumerate}
\item Construct a graph $G(V,E)$, where $V=\Xpp$ and $E = \M_d(\Xpp)$.
\item Rewire every edge $(v_i,v_j)$ for every vertex $v\in V$ with probability $\kappa$ according to $\M_r(\Xpp,g)$.
\end{enumerate}
\end{definition}

This model draws on the research in two major fields: \emph{computational geometry} and \emph{random graphs}.

Computational geometry is the study of (deterministic) algorithms for the solution of geometric problems. It often makes use of graphs. Once graphs, and therefore, the connections between vertices, are established many network properties can be determined. For example, there are a number of constructs in computational geometry that encode important aspects of the notion of nearness, such as Voronoi regions. Voronoi regions are defined in terms of nearest neighbours, their dual, the Delaunay triangularisation, minimum spanning trees, which are defined in terms of the shortest path through the vertices, and various extensions and generalisations of these ideas \citep{Marchette2004}.

Random graphs are generated using stochastic models to connect vertices. The \enquote{classical} theory of random graphs, has been established by \cite{Gilbert1959,Erdos1959,Erdos1960,Erdos1964}. Specifically, in the \emph{Gilbert model} $G(n,p)$  every possible edge occurs independently with a certain probability, while in the \emph{Erd{\H{o}}s-R{\'{e}}nyi model} $G(n,M)$ equal probabilities are assigned to all graphs with exactly $M$ edges. Many other connection models have been developed. For example, in order to restrict the degree distribution, which is not done in the generation of a classical random graph model, so-called \emph{configuration models} can be used, in which the degrees of vertices are fixed before the random graph is generated. This makes it possible to generate random networks with the desired degrees of vertices. Another example is the use of \emph{preferential attachment} mechanisms to help represent how networks change over time. Something which is not possible with \emph{static} models, such as a classical random graphs or random graphs developed using configuration models. When preferential attachment mechanisms are used the model vertices are added sequentially with a fixed number of edges connected to them.

The methods used in both of these fields both have advantages and disadvantages if they are to be used to generate spatially embedded random networks. Methods used in the field of computational geometry have the advantage that they are usually able to create edges that are not only dependent on the distance between two vertices, but also on the position of other vertices \citep{Yukich1998}. They, however, have the disadvantage that through the use of a deterministic connection algorithm, many structural characteristics, observed in real-world networks, cannot be modelled \citep{Penrose2003}. Methods used in the field of random graphs have the advantage that vertices can be connected with a certain probability, which depends both on the types of vertices and the distance between them \citep{Meester1996}. They, however, have the disadvantage that assumptions of complete independence among possible edges are largely untenable in practice \citep{Kolaczyk2009}.

In order to overcome the shortcomings of both, deterministic and random connection models, a hybrid connection model is introduced in this work (see Definition~\ref{def:hybrid_connection_model}) and a generic connection function.

Let $g$ be a connection function, that is a Borel measurable function from $[0,\infty) \to [0,1]$. In the common random connection models, two distinct vertices $u,v\in V$ are connected with probability $g(\dist(u,v))$. In other words the probability of assigning an edge to a pair of vertices depends only on their Euclidean distance to each other. However, in the context of spatially embedded random networks the probability of connecting two edges might not only depend on the distance but also on some spatial properties. Therefore, the connection function can be generalised.

\begin{theorem}
\label{the:connection_function}
  Let $g : [0,\infty) \to [0,1]$ be a connection function, $z : U \subseteq \R^d \to \R$ some scalar field\index{scalar field}, and any edge $\vec{e} : [u,v] \to \R^d \; \forall u,v\in V$ a piecewise smooth curve $\C\subset U$. Two distinct vertices $u,v\in V$ are connected with probability $g(I_{e_i})$, with 
\begin{equation}
  I_{e_i}= \int_{e_i} z \d s := \int_u^v z(\vec{e}(t))||\dot{\vec{e}}(t)|| \d t,
\end{equation}
where $\vec{e}(t)$ is a parametrization of curve $\C$, with derivative $\dot{\vec{e}} (t)$.
\end{theorem}

\begin{proof}
  Assuming an edge is a straight line between two vertices, it is trivial to show that if $z\equiv 1 \Rightarrow I_{e_i}=\dist(u,v)$. i.e. no additional (spatial) properties influencing the connection probability.
\begin{equation}
  \int_{e_i} \d s = \int_u^v ||\dot{\vec{e}}(t)|| \d t = I_{e_i} = \dist(u,v).
\end{equation}
\end{proof}

\subsection{Network generation}
\label{sec:23}

The spatially embedded random network is to be generated using the following three steps.

\begin{algorithm} :
\begin{description}[leftmargin=1.15cm]
\item[Step 1] Select an inhomogeneous Poisson point process model with an intensity function that depends on an observed covariate. This includes the selection of scaling parameters for all covariates.
\item[Step 2] Create the spatially embedded vertices, which include the selection of the number of simulated vertices  $N$.
\item[Step 3] Create the spatially embedded edges, which includes the selection of the values of the parameter $\kappa$. 
\end{description}
\end{algorithm}

If possible, the values of the parameters should be selected so that they correspond with real world data. If this is not possible, e.g. there is no readily available data, than data for similar areas may be used to together with expert opinion.

\section{Example}
\label{sec:4}

\subsection{Background}
\label{sec:41}
To illustrate the ability of the proposed methodology to be used to gain insight into a real world network, it is used to estimate standard structural statistics for part of the Swiss road network, and these are then compared with the known values.

The part of the Swiss road network (Fig.~\ref{fig:1}) that is investigated is located around the city of Chur, the capital of Grisons (or Graub\"unden in German) the largest and easternmost canton of Switzerland. Grisons is a mountainous area and includes parts of both the Rhine and Inn river valleys. Forty-one per cent of the population of Grisons live at altitudes above 1'000 MASL. The highest mountain is the Piz Bernina at 4'049 m, and the lowest point is the border with Ticino at 260 m. The Grisons' road network is comprised of circa 163 km national roads, 597 km main roads and 835 km minor roads. The canton is crossed in a north-south direction from the A13 motorway. The investigated area of the Grisons includes the districts of Imboden and the norther part of Plessur. This area is located next to the river Rhine and contains the city of Chur and Grisons' industrial area, together with the most important transportation links in the canton.

\begin{figure}[h]
\includegraphics[width=\columnwidth]{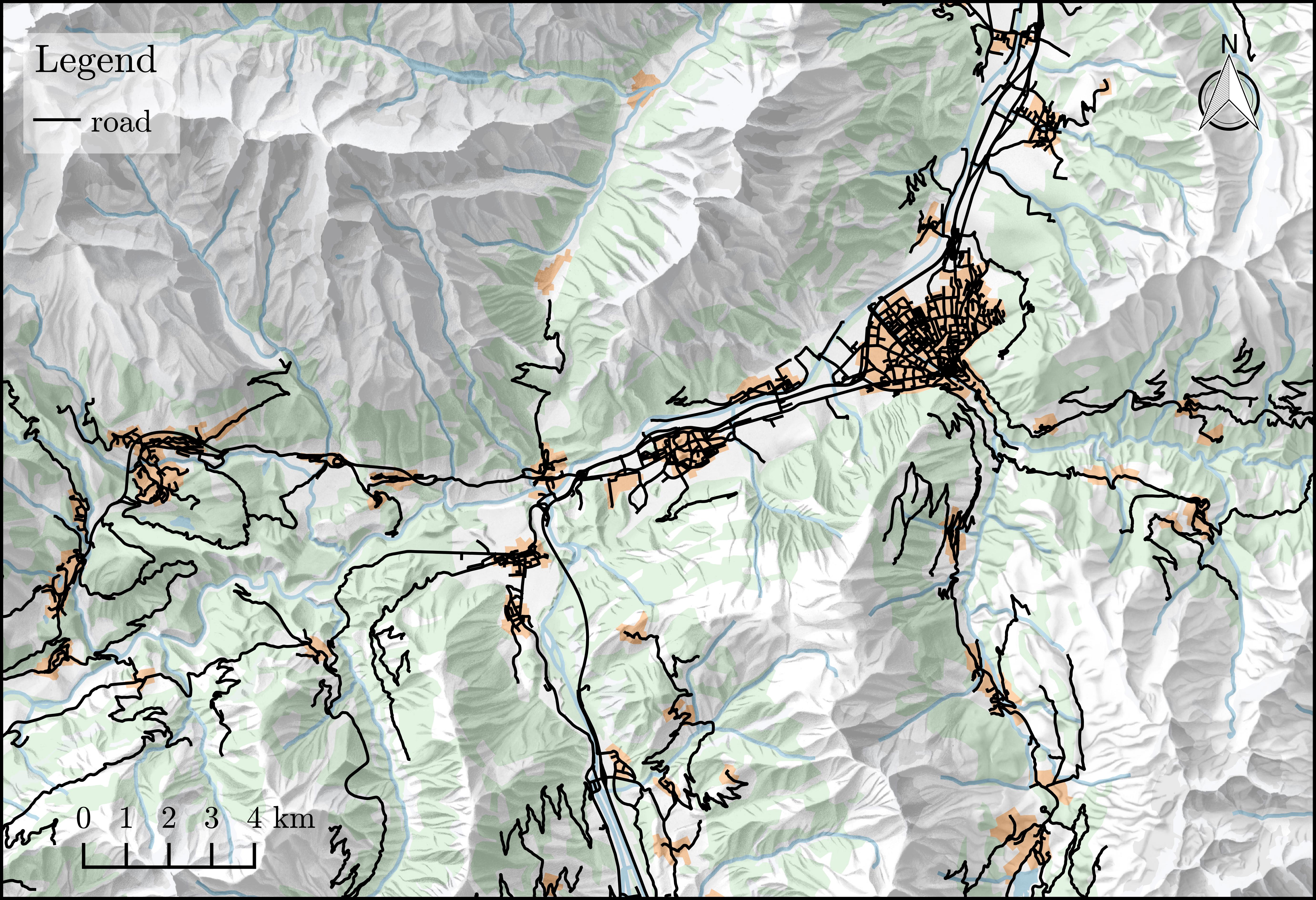}
\caption{Overview of the investigated area. (The real-world road network is shown with black lines, without distinctions between different types of roads.)}
\label{fig:1}
\end{figure}

\subsection{Description of data}
\label{sec:42}

The two types of data used to generate the spatially embedded random network were 1) terrain elevation (altitude) and 2) population density (Fig.~\ref{fig:2}). The population density was obtained from the Swiss census statistic. The data is of high enough quality that, it is possible to observe how the structures and the development of the population and households have changed over time. The data for the statistics of population and households (STATPOP) is geo-referenced, where each hectare represents a data record \citep{FSO2011}. The source for the terrain data is the RIMINI terrain model. This is a matrix model with a spatial resolution of 100 m, and covers all of Switzerland \citep{Kolbl2006}. Both databases were available as a GeoTiff for the Swiss coordinate system CH1903/LV03 LN02 (ESPG code: 21781).

The road network data was taken from the VECTOR25 data-set, provided by \cite{swisstopo2015b}. This set of data describes approximately 8.5 million objects with their position, form and its neighbourhood relations (topology), as well as the kind of object and further special attributes. VECTOR25 is composed of 9 thematic layers, one layer representing the road network. The classification of roads is based on \cite{swisstopo2011} guidelines. The roads are represented as lines. If there are level crossings then roads share the same intersection point. The VECTOR25 data-set exhibits a full national coverage in homogeneous form and quality, with an accuracy of \SIrange[range-units = single]{3}{8}{m} and it is delivered as an ESRI Shapefile for the Swiss coordinate system CH1903/LV03 LN02.

\begin{figure}[h]
\includegraphics[width=\columnwidth]{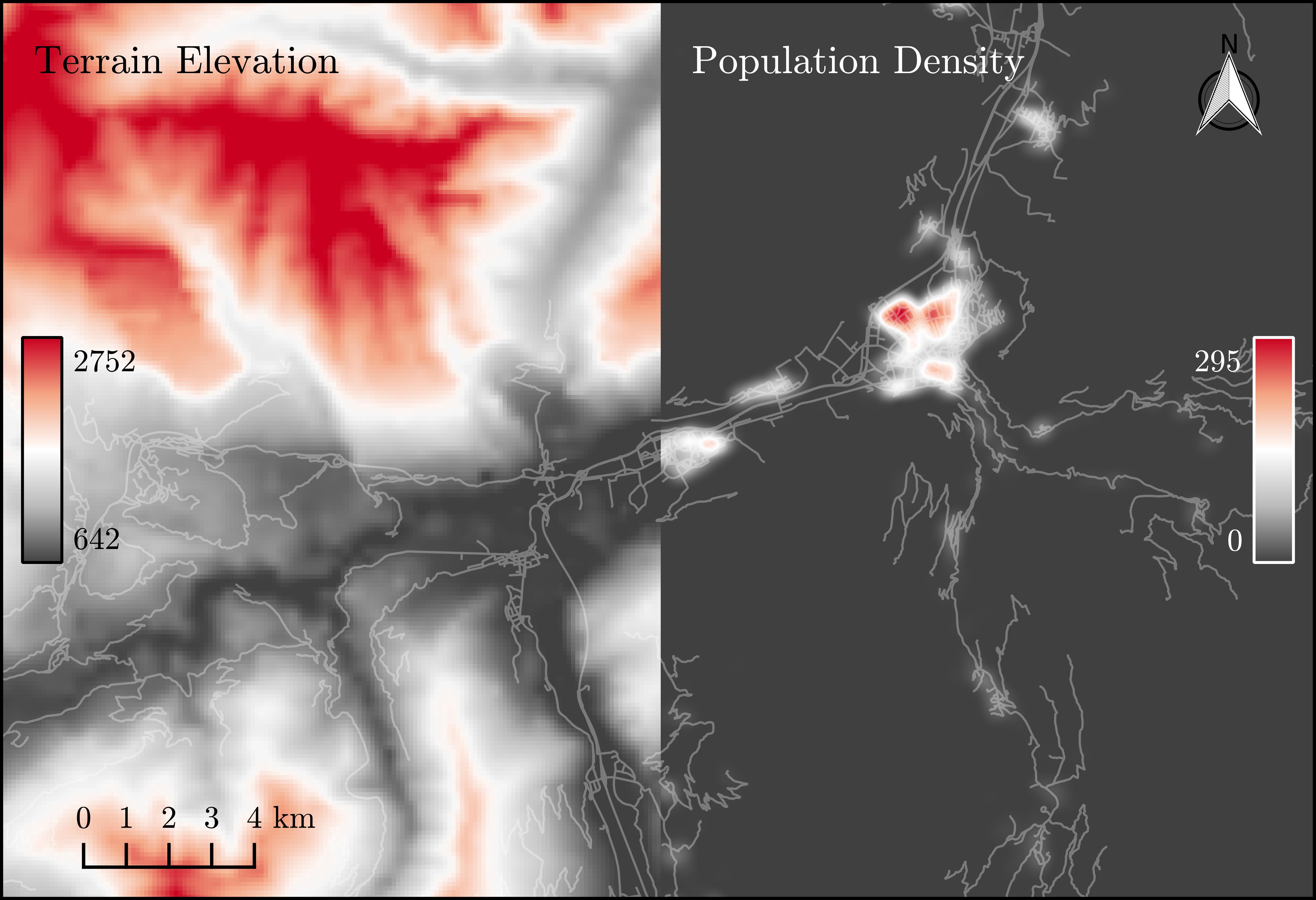}
\caption{Spatial input data for the model. On the left the terrain elevation and on the right the population density.}
\label{fig:2}
\end{figure}

\subsection{Generation of random networks}
\label{sec:43}

In this example the road intersections are modelled as vertices, while the road itself is represented by the edges. The population density is used as a covariate for the spatially embedded vertices and the terrain elevation is used as an input for the connection model. These are considered to be reasonable assumptions, since more road intersections occur in dense populated areas, while roads at high elevations are rather rare. The intensity function is, therefore, proportional to the population density:
\begin{equation}
  \lambda(u) = \beta \cdot Z(u)
\end{equation}
where $\beta$ is a scaling parameter and $Z(u)$ is the population density at $u$.

The deterministic connection model used is a \emph{relative neighbourhood graph}. It was proposed by \cite{Toussaint1980} and accounts for the \emph{relative closeness} between points. The random connection model used is a classical random connection model, i.e. edges are randomly added or removed to or from the graph, making the probability of existence of each edge indirectly proportional to the accumulated altitude over the edges. Hence, as the altitude difference between two vertices increases, the probability that they are connected decreases.

The values of the parameters were determined using data from the real-road network and a Markov Chain Monte Carlo simulation, which minimised the standard error of: the average clustering coefficient, the average shortest path length, the degree assortativity, the diameter, the number of edges and vertices, and the total length of the road network. The values of the real network are given in Tab.~\ref{tab:1} and the estimated parameters are given in Tab.~\ref{tab:2}.

\begin{table}
\caption{Properties of the real world network.}
\label{tab:1}
\small
\begin{tabularx}{\columnwidth}{Xr}
\hline%\noalign{\smallskip}
\textbf{Parameter} & \textbf{Value}\\
\hline
 Number of vertices & 1329 \\
 Number of edges & 1892 \\
 Total length of the road network in km & 711 \\
 Diameter & 63 \\
 Average shortest path length & 23.64 \\
 Average clustering coefficient & 0.079 \\
 Degree assortativity & -0.143 \\
\hline
\end{tabularx}
\end{table}

\begin{table}
\caption{Fitted model parameters.}
\label{tab:2}
\small
\begin{tabularx}{\columnwidth}{cXcc}
\hline
\textbf{Sym.} & \textbf{Parameter} & \textbf{Value} & \textbf{Range}$^a$ \\\hline
$\beta$ & Scaling parameter for the covariate of population density & $4.71$ & $(+2.96,-2.32)$ \\
$N$ & Number of simulated vertices & $4059$ & $(+814,-922)$ \\
$\kappa$ & Hybrid connection model parameter  & $0.909$ & $(+0.0655,-0.0517)$ \\
$p$ & Rewiring probability for random connection model & $0.001$& $(+0.0429,-0.0000)$ \\
\hline
\end{tabularx}
\footnotesize$^a$ Range is based on the 90\% confidence level and observed from the MCMC simulation.
\end{table}

With the values of the parameters determined, the spatially embedded random networks were generated. One example is shown in Fig.~\ref{fig:3}.

\begin{figure}[h]
\includegraphics[width=\columnwidth]{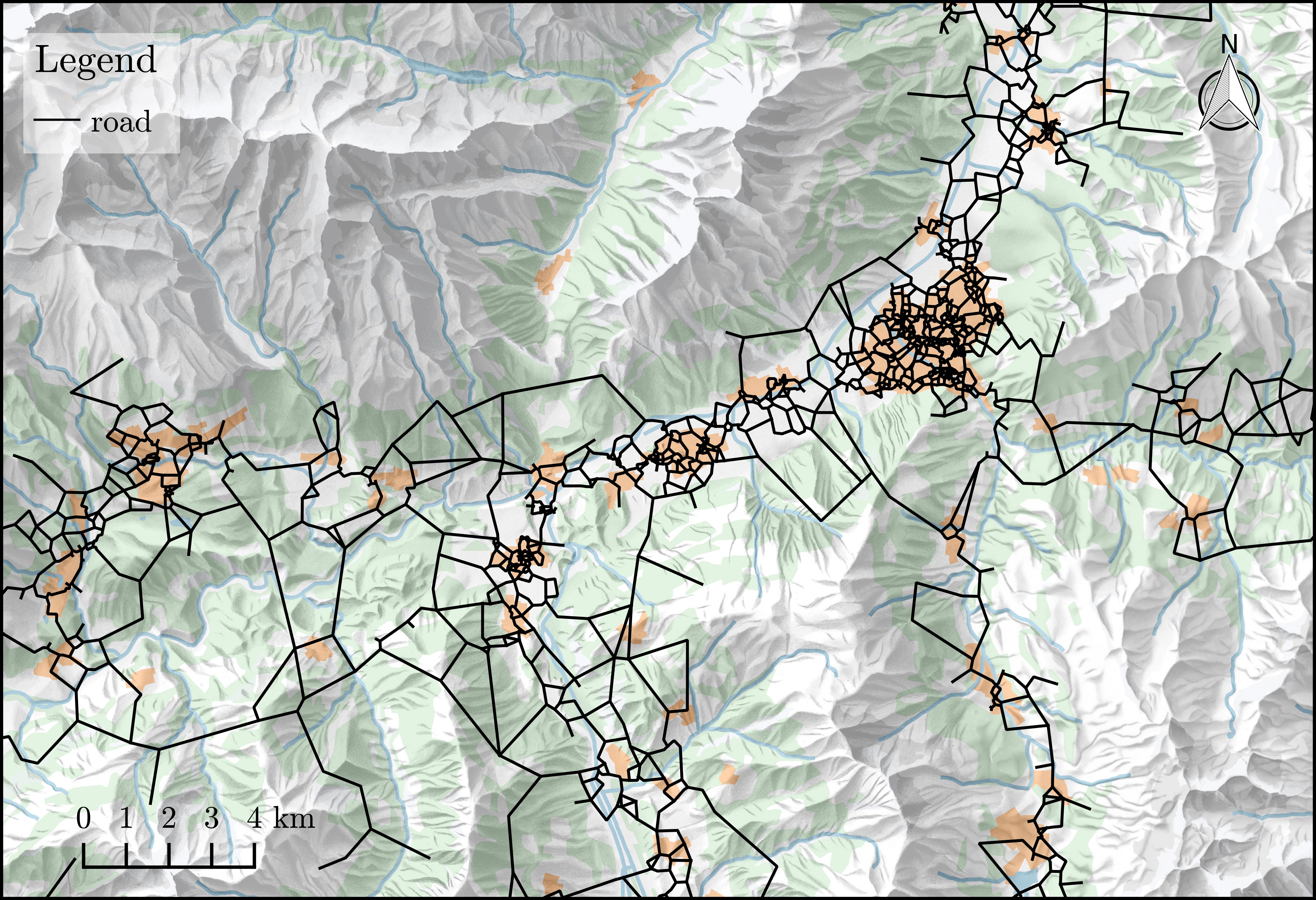}
\caption{One example of a generated spatially embedded random network.}
\label{fig:3}
\end{figure}

\subsection{Illustration of usefulness}
\label{sec:44}

In order to illustrate the ability of the proposed methodology to be used to gain insight into a real world network, the relationship between the size of the largest connected cluster, measured in terms of number of edges, and the ratio of the number of failed edges in the network and the total number of edges in the network, i.e. the failure rate, which is also sometimes taken as a measure of robustness \citep{Schneider2011}. 

The estimation is done by performing a percolation study, which describes how a network transitions from connected to disconnected \citep{Li2015}. The percolation study was done by removing edges of the network both randomly and systematically for both the random networks and the real network. When removed randomly, an edge was selected at random and removed, thereby the probability of selecting an edge is uniform over all edges in the network. When removed systematically they were removed in descending order of their betweenness centrality, which was calculated as
\begin{equation}
  c_B(v)=\sum_{s,t\in V}\frac{\sigma(s,t|e)}{\sigma{s,t}}
\end{equation}
where $V$ is the set of vertices, $\sigma(s,t)$ is the number of shortest $(s,t)$-paths, and $\sigma(s,t|e)$ is the number of those paths passing through edge $e$.

15'000 realisations of the random network were used in both cases. The results are shown in Fig.~\ref{fig:4}.

\begin{figure}[h]
\includegraphics[width=\columnwidth]{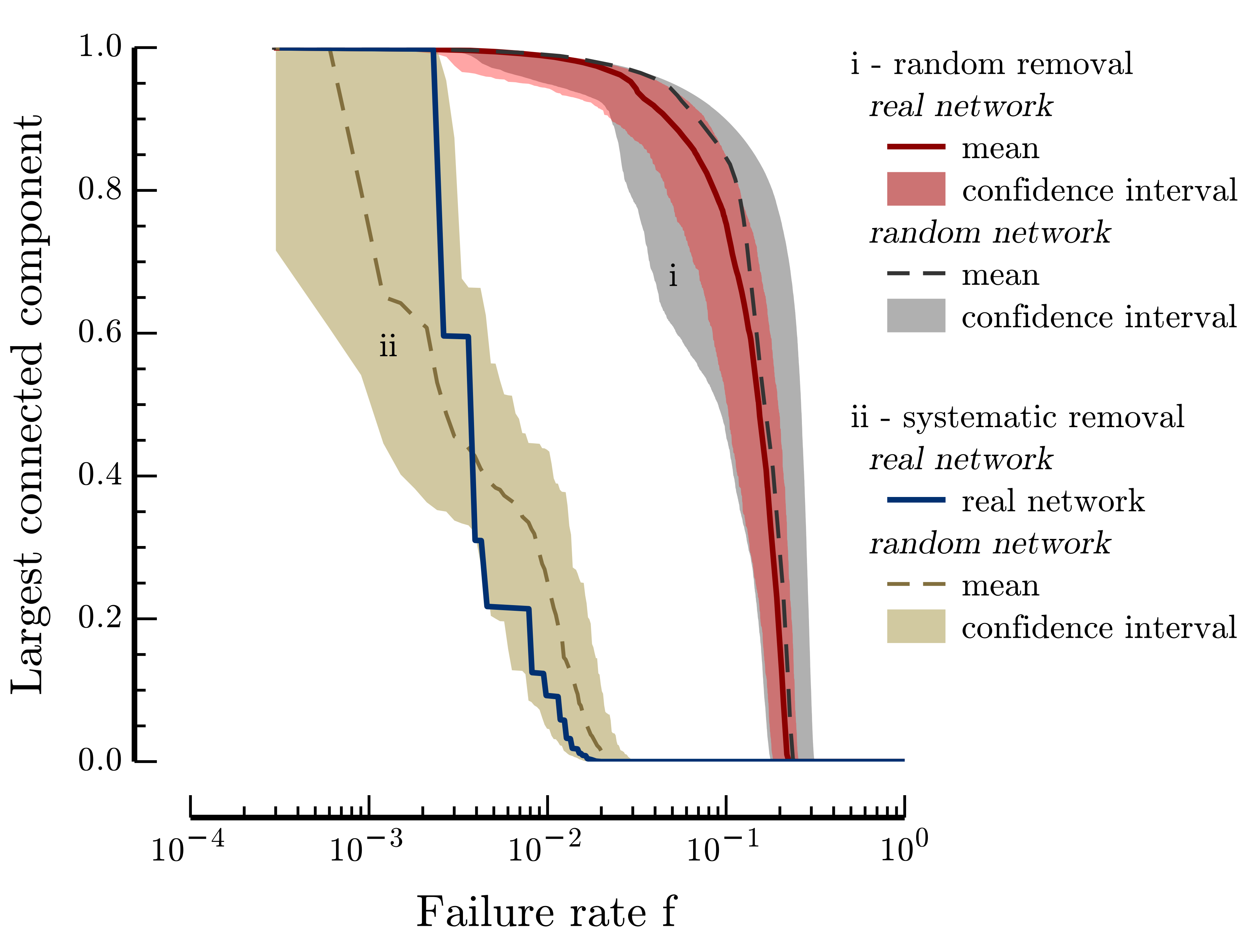}
\caption{Relationship between the size of the largest connected part of the network and the failure rate within a confidence interval of $2\sigma$.}
\label{fig:4}       % Give a unique label
\end{figure}

A failure rate $f=0$ indicates that no edges have been removed, and $f=1$ indicates that all edges have been removed. In general, the higher the $f$ the smaller the remaining parts of the network that are still connected with each other. When the network contains a cluster then it percolates and the point at which the percolation transition occurs is called the percolation threshold. Thus, the presence of the largest connected cluster is an indicator of a network that is at least partly performing its intended function. Whereas, the size of the largest connected cluster indicates how much of the network is working \citep{Newman2010}.

The curves (i) on the right-hand side in Fig.~\ref{fig:4} represent the results of a random removal of edges. The curves (ii) on the left-hand side in Fig.~\ref{fig:4} show the systematic removal of edges based on their betweenness centrality.

The values obtained from the simulations using random removal on the random network overestimate those using random removal on the real network, i.e. the real network has a lower robustness, than the random network. 

The values using the real network, are, however, within the $\sigma 2$-confidence interval of the simulation, except at the beginning of the percolation process. This might be caused by the finite amount of simulations performed. The confidence interval was generated from estimating the relationship for all 15'000 random networks.

The values obtained from the simulations using systematic removal on the random network underestimates those using systematic removal on the real network until the largest connected component reaches half their initial size. Afterwards, the random approach overestimates the real behaviour, i.e. the random network has at the beginning of the percolation process a lower robustness than the real network and in the end a higher robustness. This is caused by the choice of the connection model. While the network is densely connected in high populated areas, only few connections are created between those areas. 

The values using the real network, are, however, within the $\sigma 2$-confidence interval of the simulation for the whole percolation process, as shown in Fig.~\ref{fig:4}. Obviously, only one realisation for the real network was possible, because the topology is not changing. While for the random network several realisations can be simulated, resulting again in a confidence interval.

\section{Discussion}
\label{sec:4}

As illustrated in the example, with only few parameters, a random network can be created which mimics the properties of a real network relatively accurately. Here, only four parameters (Tab.~\ref{tab:2}) and two spatial maps (Fig.~\ref{fig:2}) are necessary to generate a spatially embedded random transportation network. In other words, only very little information on the target network is needed.

The analysis using random removal shows that both the real network and the random networks behave similarly. The largest connected cluster decreases continuously and at a failure rate of $22.26 \pm 3.26 \%$ no cluster can be observed at the real network and at failure rate of $23.94^{+7.42}_{-6.02}\%$ at the random network. Both numbers indicate that the network is not densely connected which can also be observed in Fig.~\ref{fig:1} and \ref{fig:3}.

The analysis using systematic removal shows that both the real network and the random networks have relatively large increases in failure rate with a removal of a small number edges. For example, removing systematically $1.93\%$ edges of the real network or $2.41^{+0.75}_{-0.64}\%$ of the random network leads to a disappearance of the largest connected cluster (Fig.~\ref{fig:4}).

This behaviour is due to the specific spatial properties of the road network and how the edges are removed. In other words, there are only few edges are used to connect different clusters e.g. here cities and villages, and they are removed using betweenness centrality when removes these edges first. It may be a closer fit if edges were removed taking into consideration for the traffic flow across the network. This may also be better in the estimation of risk related to the network.

The differences between the real and the random networks, although not large, may be even smaller if: 1) the number of simulations performed was increased; 2) connection models the better accounted for the real network topology were used; 3) the uncertainty in the parameters used was reduced; or 4) spatial input parameters in addition to the population density and the elevation model were used to capture more properties of the target road network.

\section{Conclusions}
\label{sec:5}

In this work a methodology to develop models of spatially embedded random networks is proposed. It combines a spatial non-homogeneous point process for node creation, which accounts for the spatial distribution of nodes, considering clustering effects of the network and a hybrid connection model for the edge creation. In contrast to many other methodologies, this methodology takes into account spatial restriction.

It was shown that the proposed methodology could be used to develop models to estimate reasonably well the relationship between the largest connected part of a part of the Swiss road network and the ratio between the number of failed links and the total number of links. This was demonstrated by conducting percolation studies with both random and systematic removal of edges.

In the illustrated example the input parameters were population density and a terrain model, and the random network models were calibrated using the seven parameters shown in Tab.~\ref{tab:1}.

It could be shown that the behaviour of the randomly created networks is similar to that of the real network. It is suspected that even more similar behaviour is possible.

\section{Future work}
\label{sec:6}
The work presented here is a useful first step towards the improved understanding of real-world spatial networks, as it enables the generation of different realisations of the same network. This is useful to take into consideration uncertainties in network properties, in situations where it is too difficult to model the network directly. Until now, uncertainties in the network topology have mainly been avoided, and opens up a myriad of possibilities in network analysis. Three important overlapping areas where it is suspected that this methodology has substantial use are:
\begin{enumerate}[label={(\roman*)},noitemsep]
\item to assess risk on complex spatially distributed networks, 
\item to model networks where there is a limited amount of information available, and 
\item to model how networks might change over time in the future. 
\end{enumerate}
Future work is to be focused on: 
\begin{itemize}[noitemsep]
\item the development of model of multiple interdependent networks
\item the use of random networks to estimate risk taking into consideration the probabilities of occurrence of loading of the network and the consequences of the possible network states.
\end{itemize}

\bibliography{literature} 

\begin{thebibliography}{21}%
\makeatletter
\providecommand \@ifxundefined [1]{%
 \@ifx{#1\undefined}
}%
\providecommand \@ifnum [1]{%
 \ifnum #1\expandafter \@firstoftwo
 \else \expandafter \@secondoftwo
 \fi
}%
\providecommand \@ifx [1]{%
 \ifx #1\expandafter \@firstoftwo
 \else \expandafter \@secondoftwo
 \fi
}%
\providecommand \natexlab [1]{#1}%
\providecommand \enquote  [1]{``#1''}%
\providecommand \bibnamefont  [1]{#1}%
\providecommand \bibfnamefont [1]{#1}%
\providecommand \citenamefont [1]{#1}%
\providecommand \href@noop [0]{\@secondoftwo}%
\providecommand \href [0]{\begingroup \@sanitize@url \@href}%
\providecommand \@href[1]{\@@startlink{#1}\@@href}%
\providecommand \@@href[1]{\endgroup#1\@@endlink}%
\providecommand \@sanitize@url [0]{\catcode `\\12\catcode `\$12\catcode
  `\&12\catcode `\#12\catcode `\^12\catcode `\_12\catcode `\%12\relax}%
\providecommand \@@startlink[1]{}%
\providecommand \@@endlink[0]{}%
\providecommand \url  [0]{\begingroup\@sanitize@url \@url }%
\providecommand \@url [1]{\endgroup\@href {#1}{\urlprefix }}%
\providecommand \urlprefix  [0]{URL }%
\providecommand \Eprint [0]{\href }%
\providecommand \doibase [0]{http://dx.doi.org/}%
\providecommand \selectlanguage [0]{\@gobble}%
\providecommand \bibinfo  [0]{\@secondoftwo}%
\providecommand \bibfield  [0]{\@secondoftwo}%
\providecommand \translation [1]{[#1]}%
\providecommand \BibitemOpen [0]{}%
\providecommand \bibitemStop [0]{}%
\providecommand \bibitemNoStop [0]{.\EOS\space}%
\providecommand \EOS [0]{\spacefactor3000\relax}%
\providecommand \BibitemShut  [1]{\csname bibitem#1\endcsname}%
\let\auto@bib@innerbib\@empty
%</preamble>
\bibitem [{\citenamefont {Vespignani}(2010)}]{Vespignani2010}%
  \BibitemOpen
  \bibfield  {author} {\bibinfo {author} {\bibfnamefont {A.}~\bibnamefont
  {Vespignani}},\ }\href@noop {} {\bibfield  {journal} {\bibinfo  {journal}
  {Nature}\ }\textbf {\bibinfo {volume} {464}},\ \bibinfo {pages} {984}
  (\bibinfo {year} {2010})}\BibitemShut {NoStop}%
\bibitem [{\citenamefont {Schneider}\ \emph {et~al.}(2011)\citenamefont
  {Schneider}, \citenamefont {Moreira}, \citenamefont {Andrade}, \citenamefont
  {Havlin},\ and\ \citenamefont {Herrmann}}]{Schneider2011}%
  \BibitemOpen
  \bibfield  {author} {\bibinfo {author} {\bibfnamefont {C.~M.}\ \bibnamefont
  {Schneider}}, \bibinfo {author} {\bibfnamefont {A.~A.}\ \bibnamefont
  {Moreira}}, \bibinfo {author} {\bibfnamefont {J.~S.}\ \bibnamefont
  {Andrade}}, \bibinfo {author} {\bibfnamefont {S.}~\bibnamefont {Havlin}}, \
  and\ \bibinfo {author} {\bibfnamefont {H.~J.}\ \bibnamefont {Herrmann}},\
  }\href@noop {} {\bibfield  {journal} {\bibinfo  {journal} {Proc. Natl. Acad.
  Sci.}\ }\textbf {\bibinfo {volume} {108}},\ \bibinfo {pages} {3838} (\bibinfo
  {year} {2011})}\BibitemShut {NoStop}%
\bibitem [{\citenamefont {Li}\ \emph {et~al.}(2015)\citenamefont {Li},
  \citenamefont {Zhang}, \citenamefont {Zio}, \citenamefont {Havlin},\ and\
  \citenamefont {Kang}}]{Li2015}%
  \BibitemOpen
  \bibfield  {author} {\bibinfo {author} {\bibfnamefont {D.}~\bibnamefont
  {Li}}, \bibinfo {author} {\bibfnamefont {Q.}~\bibnamefont {Zhang}}, \bibinfo
  {author} {\bibfnamefont {E.}~\bibnamefont {Zio}}, \bibinfo {author}
  {\bibfnamefont {S.}~\bibnamefont {Havlin}}, \ and\ \bibinfo {author}
  {\bibfnamefont {R.}~\bibnamefont {Kang}},\ }\href {\doibase
  10.1016/j.ress.2015.05.021} {\bibfield  {journal} {\bibinfo  {journal}
  {Reliability Engineering \& System Safety}\ }\textbf {\bibinfo {volume}
  {142}},\ \bibinfo {pages} {556 } (\bibinfo {year} {2015})}\BibitemShut
  {NoStop}%
\bibitem [{\citenamefont {Kr{\"o}ger}\ and\ \citenamefont
  {Zio}(2011)}]{Kroger2011}%
  \BibitemOpen
  \bibfield  {author} {\bibinfo {author} {\bibfnamefont {W.}~\bibnamefont
  {Kr{\"o}ger}}\ and\ \bibinfo {author} {\bibfnamefont {E.}~\bibnamefont
  {Zio}},\ }\href@noop {} {\emph {\bibinfo {title} {Vulnerable Systems}}},\
  SpringerLink : B{\"u}cher\ (\bibinfo  {publisher} {Springer},\ \bibinfo
  {year} {2011})\BibitemShut {NoStop}%
\bibitem [{\citenamefont {Cadini}\ \emph {et~al.}(2015)\citenamefont {Cadini},
  \citenamefont {Azzolin},\ and\ \citenamefont {Zio}}]{Cadini2015}%
  \BibitemOpen
  \bibfield  {author} {\bibinfo {author} {\bibfnamefont {F.}~\bibnamefont
  {Cadini}}, \bibinfo {author} {\bibfnamefont {A.}~\bibnamefont {Azzolin}}, \
  and\ \bibinfo {author} {\bibfnamefont {E.}~\bibnamefont {Zio}},\ }in\ \href
  {\doibase 10.1201/b19094-578} {\emph {\bibinfo {booktitle} {Safety and
  Reliability of Complex Engineered Systems, Proceedings of the 25th European
  Safety and Reliability Conference (ESREL)}}}\ (\bibinfo  {publisher} {CRC
  Press},\ \bibinfo {address} {Zurich, Switzerland},\ \bibinfo {year}
  {2015})\BibitemShut {NoStop}%
\bibitem [{\citenamefont {Barnett}\ \emph {et~al.}(2007)\citenamefont
  {Barnett}, \citenamefont {Di~Paolo},\ and\ \citenamefont
  {Bullock}}]{Barnett2007}%
  \BibitemOpen
  \bibfield  {author} {\bibinfo {author} {\bibfnamefont {L.}~\bibnamefont
  {Barnett}}, \bibinfo {author} {\bibfnamefont {E.}~\bibnamefont {Di~Paolo}}, \
  and\ \bibinfo {author} {\bibfnamefont {S.}~\bibnamefont {Bullock}},\ }\href
  {\doibase 10.1103/PhysRevE.76.056115} {\bibfield  {journal} {\bibinfo
  {journal} {Phys. Rev. E}\ }\textbf {\bibinfo {volume} {76}},\ \bibinfo
  {pages} {056115} (\bibinfo {year} {2007})}\BibitemShut {NoStop}%
\bibitem [{\citenamefont {Marchette}(2004)}]{Marchette2004}%
  \BibitemOpen
  \bibfield  {author} {\bibinfo {author} {\bibfnamefont {D.}~\bibnamefont
  {Marchette}},\ }\href@noop {} {\emph {\bibinfo {title} {Random Graphs for
  Statistical Pattern Recognition}}},\ Wiley Series in Probability and
  Statistics\ (\bibinfo  {publisher} {Wiley},\ \bibinfo {year}
  {2004})\BibitemShut {NoStop}%
\bibitem [{\citenamefont {Gilbert}(1959)}]{Gilbert1959}%
  \BibitemOpen
  \bibfield  {author} {\bibinfo {author} {\bibfnamefont {E.~N.}\ \bibnamefont
  {Gilbert}},\ }\href {\doibase 10.1214/aoms/1177706098} {\bibfield  {journal}
  {\bibinfo  {journal} {Ann. Math. Statist.}\ }\textbf {\bibinfo {volume}
  {30}},\ \bibinfo {pages} {1141} (\bibinfo {year} {1959})}\BibitemShut
  {NoStop}%
\bibitem [{\citenamefont {Erd{\H{o}}s}\ and\ \citenamefont
  {R{\'{e}}nyi}(1959)}]{Erdos1959}%
  \BibitemOpen
  \bibfield  {author} {\bibinfo {author} {\bibfnamefont {P.}~\bibnamefont
  {Erd{\H{o}}s}}\ and\ \bibinfo {author} {\bibfnamefont {A.}~\bibnamefont
  {R{\'{e}}nyi}},\ }\href@noop {} {\bibfield  {journal} {\bibinfo  {journal}
  {Publ. Math. Debrecen}\ }\textbf {\bibinfo {volume} {6}},\ \bibinfo {pages}
  {290} (\bibinfo {year} {1959})}\BibitemShut {NoStop}%
\bibitem [{\citenamefont {Erd{\H{o}}s}\ and\ \citenamefont
  {R{\'{e}}nyi}(1960)}]{Erdos1960}%
  \BibitemOpen
  \bibfield  {author} {\bibinfo {author} {\bibfnamefont {P.}~\bibnamefont
  {Erd{\H{o}}s}}\ and\ \bibinfo {author} {\bibfnamefont {A.}~\bibnamefont
  {R{\'{e}}nyi}},\ }in\ \href@noop {} {\emph {\bibinfo {booktitle} {Publ. Math.
  Inst. Hung. Acad. Sci.}}}\ (\bibinfo {year} {1960})\ pp.\ \bibinfo {pages}
  {17--61}\BibitemShut {NoStop}%
\bibitem [{\citenamefont {Erd{\H{o}}s}\ and\ \citenamefont
  {R{\'{e}}nyi}(1964)}]{Erdos1964}%
  \BibitemOpen
  \bibfield  {author} {\bibinfo {author} {\bibfnamefont {P.}~\bibnamefont
  {Erd{\H{o}}s}}\ and\ \bibinfo {author} {\bibfnamefont {A.}~\bibnamefont
  {R{\'{e}}nyi}},\ }\href {\doibase 10.1007/BF02066689} {\bibfield  {journal}
  {\bibinfo  {journal} {Acta Mathematica Academiae Scientiarum Hungarica}\
  }\textbf {\bibinfo {volume} {12}},\ \bibinfo {pages} {261} (\bibinfo {year}
  {1964})}\BibitemShut {NoStop}%
\bibitem [{\citenamefont {Yukich}(1998)}]{Yukich1998}%
  \BibitemOpen
  \bibfield  {author} {\bibinfo {author} {\bibfnamefont {J.~E.}\ \bibnamefont
  {Yukich}},\ }\href@noop {} {\emph {\bibinfo {title} {{Probability theory of
  classical Euclidean optimization problems}}}},\ Lecture Notes in Mathematics\
  (\bibinfo  {publisher} {Springer},\ \bibinfo {address} {Berlin},\ \bibinfo
  {year} {1998})\BibitemShut {NoStop}%
\bibitem [{\citenamefont {Penrose}(2003)}]{Penrose2003}%
  \BibitemOpen
  \bibfield  {author} {\bibinfo {author} {\bibfnamefont {M.~D.}\ \bibnamefont
  {Penrose}},\ }\href@noop {} {\emph {\bibinfo {title} {Random Geometric
  Graphs}}},\ Oxford scholarship online\ (\bibinfo  {publisher} {Oxford
  University Press},\ \bibinfo {year} {2003})\BibitemShut {NoStop}%
\bibitem [{\citenamefont {Meester}\ and\ \citenamefont
  {Roy}(1996)}]{Meester1996}%
  \BibitemOpen
  \bibfield  {author} {\bibinfo {author} {\bibfnamefont {R.}~\bibnamefont
  {Meester}}\ and\ \bibinfo {author} {\bibfnamefont {R.}~\bibnamefont {Roy}},\
  }\href@noop {} {\emph {\bibinfo {title} {Continuum Percolation}}},\ Cambridge
  Tracts in Mathematics\ (\bibinfo  {publisher} {Cambridge University Press},\
  \bibinfo {year} {1996})\BibitemShut {NoStop}%
\bibitem [{\citenamefont {Kolaczyk}(2009)}]{Kolaczyk2009}%
  \BibitemOpen
  \bibfield  {author} {\bibinfo {author} {\bibfnamefont {E.~D.}\ \bibnamefont
  {Kolaczyk}},\ }\href@noop {} {\emph {\bibinfo {title} {Statistical Analysis
  of Network Data: Methods and Models}}},\ \bibinfo {edition} {1st}\ ed.\
  (\bibinfo  {publisher} {Springer Publishing Company, Incorporated},\ \bibinfo
  {year} {2009})\BibitemShut {NoStop}%
\bibitem [{\citenamefont {{Federal Statistical Office}}(2011)}]{FSO2011}%
  \BibitemOpen
  \bibfield  {author} {\bibinfo {author} {\bibnamefont {{Federal Statistical
  Office}}},\ }\href@noop {} {\emph {\bibinfo {title} {The new census}}},\
  \bibinfo {type} {Tech. Rep.}\ \bibinfo {number} {1143-1101-05}\ (\bibinfo
  {institution} {Federal Statistical Office (FSO)},\ \bibinfo {address}
  {Neuch\^atel, Swizerland},\ \bibinfo {year} {2011})\BibitemShut {NoStop}%
\bibitem [{\citenamefont {K\"olbl}(2006)}]{Kolbl2006}%
  \BibitemOpen
  \bibfield  {author} {\bibinfo {author} {\bibfnamefont {O.}~\bibnamefont
  {K\"olbl}},\ }\href@noop {} {\emph {\bibinfo {title} {Gel\"andedaten}}},\
  \bibinfo {type} {Tech. Rep.}\ \bibinfo {number} {B.1.1}\ (\bibinfo
  {institution} {Federal Statistical Office (FSO)},\ \bibinfo {address}
  {Neuch\^atel, Swizerland},\ \bibinfo {year} {2006})\BibitemShut {NoStop}%
\bibitem [{\citenamefont {swisstopo}(2015)}]{swisstopo2015b}%
  \BibitemOpen
  \bibfield  {author} {\bibinfo {author} {\bibnamefont {swisstopo}},\ }\href
  {http://www.swisstopo.admin.ch} {\enquote {\bibinfo {title} {Vector25},}\ }
  (\bibinfo {year} {2015}),\ \bibinfo {note} {federal Office of
  Topography}\BibitemShut {NoStop}%
\bibitem [{\citenamefont {swisstopo}(2011)}]{swisstopo2011}%
  \BibitemOpen
  \bibfield  {author} {\bibinfo {author} {\bibnamefont {swisstopo}},\
  }\href@noop {} {\emph {\bibinfo {title} {Karten-Signaturen}}},\ \bibinfo
  {type} {Tech. Rep.}\ (\bibinfo  {institution} {Federal Office of
  Topography},\ \bibinfo {address} {Wabern, Switzerland},\ \bibinfo {year}
  {2011})\BibitemShut {NoStop}%
\bibitem [{\citenamefont {Toussaint}(1980)}]{Toussaint1980}%
  \BibitemOpen
  \bibfield  {author} {\bibinfo {author} {\bibfnamefont {G.~T.}\ \bibnamefont
  {Toussaint}},\ }\href {\doibase 10.1016/0031-3203(80)90066-7} {\bibfield
  {journal} {\bibinfo  {journal} {Pattern Recognition}\ }\textbf {\bibinfo
  {volume} {12}},\ \bibinfo {pages} {261 } (\bibinfo {year}
  {1980})}\BibitemShut {NoStop}%
\bibitem [{\citenamefont {Newman}(2010)}]{Newman2010}%
  \BibitemOpen
  \bibfield  {author} {\bibinfo {author} {\bibfnamefont {M.~E.~J.}\
  \bibnamefont {Newman}},\ }\href@noop {} {\emph {\bibinfo {title} {Networks:
  An Introduction}}}\ (\bibinfo  {publisher} {Oxford University Press, Inc.},\
  \bibinfo {address} {New York, NY, USA},\ \bibinfo {year} {2010})\BibitemShut
  {NoStop}%
\end{thebibliography}%
 
\end{document}